%% file: arxiv.tex
\documentclass[12pt]{article}

\pdfoutput=1

\usepackage[margin=1in]{geometry}
\usepackage{setspace}
\usepackage[
  bookmarks=true,
  bookmarksnumbered=true,
  bookmarksopen=true,
  pdfborder={0 0 0},
  breaklinks=true,
  colorlinks=true,
  linkcolor=black,
  citecolor=black,
  filecolor=black,
  urlcolor=black,
]{hyperref}

\usepackage[normalem]{ulem}
\usepackage[round]{natbib}
\usepackage{amssymb,amsmath,amsthm}

\usepackage[capitalize]{cleveref}

\theoremstyle{plain}
\newtheorem{theorem}{Theorem}[section]
\newtheorem{lemma}[theorem]{Lemma}
\newtheorem{proposition}[theorem]{Proposition}

\theoremstyle{definition}
\newtheorem{definition}[theorem]{Definition}

\usepackage{bbm}
\usepackage{nicefrac}
\usepackage{comment}
\usepackage{mathtools}

\DeclareMathOperator*{\argmax}{arg\,max}

\title{Multi-District School Choice: Playing on Several Fields\thanks{This work adapts and extends the undergraduate thesis of \citet{Yin2022}. We thank Scott Kominers and Assaf Romm for insightful comments and discussions. Gonczarowski gratefully acknowledges research support by the National Science Foundation (NSF-BSF grant No.\ 2343922), Harvard FAS Inequality in America Initiative, and Harvard FAS Dean’s Competitive Fund for Promising Scholarship. Zhang was supported by an NSF Graduate Research Fellowship.}}

\author{
 Yannai A. Gonczarowski\thanks{Department of Economics and Department of Computer Science, Harvard University.\\
 Email: yannai@gonch.name.}
 \and 
 Michael Yin\thanks{Paris School of Economics. Email: michaelyin2018@gmail.com.}
 \and
 Shirley Zhang\thanks{Department of Computer Science, Harvard University. Email: szhang2@g.harvard.edu.}
}

\begin{document}

\begin{titlepage}

\maketitle

\begin{abstract}
We extend the seminal model of \citet{pathak2008leveling} to a setting with multiple school districts, each running its own separate centralized match, and focus on the case of two districts. In our setting, in addition to each student being either sincere or sophisticated, she is also either \emph{constrained}---able to apply only to schools within her own district of residence---or \emph{unconstrained}---able to choose any single district within which to apply. We show that several key results from \citet{pathak2008leveling} no longer hold: A sophisticated student may prefer for a sincere student to become sophisticated, and a sophisticated student may prefer for her own district to use Deferred Acceptance over the Boston Mechanism, irrespective of the mechanism used by the other district. We furthermore investigate the preferences of students over the constraint levels of other students. Many of these phenomena appear abundantly in large random markets. 
\end{abstract}

\end{titlepage}

\input{01-introduction}
\input{02-model}
\input{03-soph_prefer_soph}
\input{04-soph_can_prefer_da}
\input{05-constrained}
\input{06-sincere_unconstrained}
\input{07-discussion}

\bibliographystyle{abbrvnat}
\bibliography{bib}

\newpage
\appendix
\input{08-appendix}

\end{document}

%% file: 01-introduction.tex
\section{Introduction}

The Boston Mechanism (henceforth, BM; also sometimes referred to as the ``Immediate Acceptance'' mechanism) is a widely used school-choice mechanism, especially in school-choice systems that were not (re)designed by economists or computer scientists. This mechanism first maximizes the number of applicants who get their first-choice school (breaking ties based on the priorities that students have at the different schools, i.e., based on the schools' ``preferences''); then, subject to that, maximizes the number of applicants who get their second-choice school; then, subject to that, maximizes the number of applicants who get their third-choice school; and so forth. Despite being a very natural mechanism, BM suffers from various unattractive qualities, such as not being strategyproof and resulting in unstable matchings. Due to these and other shortcomings, there has been a push since the turn of the millennium \citep[e.g.,][]{ abdulkadirouglu2005new,abdulkadirouglu2005boston,pathak2008leveling} to replace BM with the better-behaved Deferred Acceptance mechanism (\citealp{gale1962college}; henceforth, DA) in school-choice systems.

One of the most compelling arguments given in favor of replacing BM with DA is the equity argument that originates in the seminal paper of \citet{pathak2008leveling}, which considers a setting with some students being sincere (i.e., uninformed and unstrategic, always reporting their true preferences) and some being sophisticated (i.e., informed and strategic, together playing a Nash equilibrium). That paper proves that sophisticated students weakly prefer BM over DA (the latter mechanism could be seen as a baseline that treats sincere and sophisticated students equally, due to its strategyproofness). This leads \citet{pathak2008leveling} to view BM as weakly (and many times strictly) conferring an advantage to sophisticated students over sincere ones. \citet{pathak2008leveling} furthermore prove that when BM is used, sophisticated students weakly prefer for sincere students to remain sincere, giving a plausible explanation as to why informed parent groups might not be likely to share their know-how with parents outside their groups or social circles.

A school district running a centralized matching mechanism is not an isolated capsule. Many districts, each running an independent centralized match, might exist next to each other, and some students might be able to choose in which district's match to participate. For example, a 2005 report for the Berkeley Unified School District in California
\citep{fried2005attending} 
estimated that between 7.8\% and 12\% of the district’s high schoolers were ``attending [the district] unofficially,'' and actually lived out-of-district. Choosing one's school district can thus be done without official permission (as in the case above) at personal risk,\footnote{See, for instance, \citet{martin2011mother} 
for more context on this illegal phenomenon, known as ``boundary hopping'' or ``residency fraud,'' which has at times led to prison sentences for parents.} or legally by moving to that district, an option many times available only to populations with greater financial resources.\footnote{Comparing the prices of houses located near school district boundaries, \citet{black1999better} 
estimates that parents are willing to pay 2.1\% more to enroll their child
in a district with a 5\% higher mean test score, and \citet{bayer2007unified} 
similarly estimate a 1.8\%
higher willingness to pay for homes in a district with an average test score that is higher by one standard deviation.}

Neighboring school districts are often independent of one another and might use different mechanisms to match students to schools. Due to the ability of some students to choose their school districts, one district switching its mechanism has the potential to change the multi-district equilibrium, changing students' strategies not only in terms of how they rank schools within a district but also in terms of their choice of school district. In this paper, we examine the two predictions of \citet{pathak2008leveling} that we describe above in a multi-district setting in which district choice (by the students who, for instance, possess the resources to officially relocate or are willing to risk punishment) is endogenized as part of the equilibrium. That is, in our setting, in addition to each student being either sincere or sophisticated, she is also either \emph{constrained}---able to apply only to schools within her own district of residence---or \emph{unconstrained}---able to strategically choose any single district within which to apply.\footnote{In either case, if a student is sophisticated, she strategically orders her submitted preference list over the schools in the district to which she applies.}

We prove that even when considering only two school districts, both of the predictions of \citet{pathak2008leveling} that we describe above no longer hold. Specifically, a sophisticated student may strictly prefer for her district to use DA over BM, irrespective of whether she is constrained or unconstrained and of the mechanism used by the other district. Furthermore, a sophisticated student may strictly prefer for some sincere student to become sophisticated. The latter phenomenon also appears abundantly in large random markets; that is, for sufficiently large random markets, a constant fraction of sophisticated students strictly prefer that at least some sincere students become sophisticated. We round out our investigation by asking, for completeness, whether some students might prefer for others to change their constraint types, e.g., whether an unconstrained student might prefer for another student to become unconstrained, or whether a constrained student might prefer for another student who resides in a different district to become constrained. We prove a strong ``anything goes'' result showing that every possible such combination appears abundantly in large random markets.

Our results are not without limitations. For one, consider the phenomenon of sophisticated students strictly preferring DA over BM. While we show this phenomenon to be possible, it is our only result that might be rare in random markets,\footnote{We do show that its frequency at least does not diminish as the market grows.} which could still lend credence to an argument in favor of DA over BM. Importantly, though, this argument becomes a quantitative issue of relative frequency rather than a qualitative issue of existence. More broadly, our paper is not intended to advocate for the use of BM. Rather, first and foremost, it serves to introduce a formal model of multi-district school choice and provide a proof-of-concept highlighting that taking into account the broader landscape beyond only a single district may qualitatively change the analysis, including the validity of certain arguments for or against the use of various mechanisms. When designing the specifics of a mechanism or market (be it when choosing the overall mechanism as discussed in this paper, or possibly even when considering far more minute implementation details), one always weighs the specifics of the market in question. Our results highlight that in some cases, this market should be even more broadly defined than is customary.

The remainder of this paper is structured as follows. After reviewing related work, in \cref{sec:model} we present the multi-district school choice model. In \cref{sec:sophistication_types} we prove our first main result, regarding preference over sophistication types. In \cref{sec:mechanism_choice} we prove our second main result, regarding preference over mechanisms. In \cref{sec:constraint_types} we prove our results regarding preference over constraint types. 
In \cref{sec:sincere-unconstrained}, we dive deeper into the mechanics that enable our first main result---that sophisticated students might strictly prefer for some or all sincere students to become sophisticated. We uncover that there are two distinct mechanisms that can drive this result, identify the precise features of our model that enable each of these mechanisms, show their robustness, and derive necessary as well as sufficient conditions over the market structure for each of these mechanisms to manifest. We conclude with a discussion in \cref{sec:discussion}.

\subsection{Related Work}

The application of mechanism design to school choice originated in \citet{abdulkadirouglu2003school}. Strategic opportunities in BM had been observed when this mechanism was first described in the economic literature \citep{abdulkadirouglu2005boston}, and were subsequently shown in the lab \citep{chen2006school} and in the field \citep{calsamiglia2018priorities}. Welfare arguments in favor of DA over BM have appeared in \citet{ergin2006games} and \citet{kojima2008games}, culminating in the equity and fairness arguments of \citet{pathak2008leveling}. Several papers examine some of the predictions of \citet{pathak2008leveling} in various extended models (still within a single district), such as with coarse priority structures \citep{abdulkadirouglu2011resolving,babaioff2019playing} or with a finer classification of sophistication types \citep{Zhang2021}. Our large-market analysis methods are technically most closely related to those of \citet{babaioff2019playing}.

To our knowledge, ours is the first theoretical analysis of multi-district school choice where some students can choose their district. \citet{Hafalir2022} also discuss a notion of interdistrict school choice, but in their model all students can freely rank schools regardless of district, and the focus is on a policy goal of diversity across districts. \citet{Grigoryan2023} considers multiple neighborhoods with a school in each one, but families can again freely rank schools regardless of neighborhood, and the emphasis is on aggregate welfare and welfare for low-income families under different matching mechanisms. In contrast to these papers, because we utilize the lens of constrained and unconstrained students, our multi-district school choice problems are distinct from large single-district problems, allowing us to re-examine the equity arguments of \citet{pathak2008leveling} in this setting. 

Within a single district, closest to our work are previous papers that analyze different types of schools (such as charter, magnet, and private schools) coexisting with public schools. For instance, considering schools that are not district-run and can therefore choose to use their own admissions systems, \citet{ekmekci2019common} 
analyze school incentives for participating in a unified enrollment system in a single district. \citet{Dogan2019}, meanwhile, compare the efficiency of unified and divided enrollment systems in a district that has selective and nonselective schools (where students can be admitted to a school of each type under the divided enrollment system). Other papers assume that there is no option for a unified enrollment system, and instead analyze a ``slightly decentralized'' mechanism that can be used to rematch students with the vacant seats that arise from schools of several types accepting the same student
\citep{manjunath2016two,turhan2019welfare,afacan2022parallel}. \citet{akbarpour2022centralized} 
study how exogenously varying the value of outside options affects behavior and preference over mechanisms. Although these threads of research bear some resemblance to the multi-district school choice problem that we study, there are a number of key differences. \citet{ekmekci2019common} fix the students in the district and focus on the choice by schools of whether to participate in a unified enrollment system; we instead fix the schools in each of multiple districts and allow
some students to choose in which district to apply. Meanwhile, \citet{manjunath2016two}, \citet{turhan2019welfare}, and \citet{afacan2022parallel} allow every student to report rankings for each school type and choose between matches they receive; the same is true for Doğan and Yenmez (2019) under divided enrollment. \citet{akbarpour2022centralized} similarly do not require students to forego their outside option in order to participate in a match. By contrast, we distinguish between students who can and cannot utilize district choice, and even those who can are only able to rank schools in their single chosen district (and therefore only receive one match, and cannot keep a guaranteed outside option). Our paper centers on how students endogenously choose their district of enrollment and thus affect the landscape of a school choice problem, a modeling decision that distinguishes this paper from the above prior work.

Finally, our investigation into the interplay between the choice of mechanism for one district and the multi-district equilibrium can be seen as contributing to a recent line of work on ``partial mechanism design'' (e.g., \citealp{PhilipponS2012,Tirole2012,Kang2023}; see \citealp{KangM2023tutorial}, for a review).

%% file: 02-model.tex
\section{Model}\label{sec:model}

\subsection{Standard Concepts}

Employing much of the notation of \citet{pathak2008leveling}, we use the following standard concepts from the school choice literature. 

\subsubsection{Single-District School Choice}

In a (single-district) \emph{school choice problem}, there is a set of \emph{students} $I = \{i_1, ..., i_n\}$ and a set of \emph{schools} $S = \{s_1, ..., s_m\}$. Each student $i$ has a strict \emph{preference ordering} $P_i$ over some subset of $S$, and $i$ prefers remaining unassigned over being assigned to schools that are not in this subset. Each school $s$ has a \emph{capacity} of $q_s$ seats, which is the maximum number of students that $s$ can accept, and a strict \emph{priority ordering} $\pi_s$ over all students. The schools' priorities for students are \emph{responsive} in the sense that:
\begin{itemize}
    \item a school cannot reject students if it is not at capacity, and
    \item a school cannot accept a lower priority student over a higher priority student, regardless of which other students may or may not be accepted.
\end{itemize}
School priority orderings and capacities are public (e.g., set by policy). So that school assignments can be determined, each student submits a \emph{rank-order list} (henceforth, \emph{ROL}) of any number of schools, which may or may not match her actual preference ordering.

\subsubsection{Mechanisms}

A school choice mechanism uses students' submitted ROLs and schools' priority orderings and capacities to determine school assignments in a single district. Two such mechanisms are the \emph{Boston Mechanism} (abbreviated as \emph{BM}) and Deferred Acceptance (abbreviated as~\emph{DA}).

\begin{definition}
The Boston Mechanism (BM) \citep{abdulkadirouglu2005boston} operates in several rounds as follows: 
\begin{itemize}
\item Round 1: Each student who submitted a non-empty ROL applies to the school she ranked 1st on her ROL. For each school, if there are at least as many seats available as applicants, the school (permanently) accepts every applicant. Otherwise, each school (permanently) allocates seats to applicants based on the school’s priority ordering up to its capacity, and rejects the remaining students for whom no seats remain.
\item Round $k > 1$: Consider only students who have not yet been accepted to a school (i.e., the students who have been rejected by the schools $1$st through $(k\!-\!1)$th on their ROLs). Each student who submitted an ROL of at least length $k$ applies to the school she ranked $k$th. For each school, applicants are accepted or rejected in the same way as in Round 1, where the seats available are those that were not already filled in previous rounds. If a school has no seats available at the beginning of the round, it rejects all new applicants.
\item This process terminates when every student has been either assigned a seat at a school or rejected by every school on her ROL, in which case she remains unassigned.
\end{itemize}
\end{definition}

Observe that a student is immediately permanently accepted or rejected when she applies to a school in BM. For this reason, BM is also known as the Immediate Acceptance mechanism.

\begin{definition}
The Deferred Acceptance mechanism \citep{gale1962college} also operates in several rounds, but with only tentative acceptances until the very end, as follows: 
\begin{itemize}
\item Round 1: Each student who submitted a non-empty ROL applies to the school she ranked 1st on her ROL. For each school, if there are at least as many seats available as applicants, the school tentatively accepts every applicant. Otherwise, each school tentatively allocates seats to applicants based on the school’s priority ordering up to its capacity, and (permanently) rejects the remaining students for whom no seats remain.
\item Round $k > 1$: Consider only students who are not currently tentatively accepted at a school (i.e., the students who were rejected by a school in round $k\!-\!1$). Each of these students applies to the school highest on her ROL that has not already rejected her. For each school, new applicants are considered alongside tentatively accepted students. All of these students are compared based on the school's priority ordering and are tentatively accepted or permanently rejected in the same way as in Round 1, where all seats at the school are initially considered available.\footnote{Students who were tentatively accepted by the school in round $k\!-\!1$ are not conferred any advantage, and may still be permanently rejected by that school in round $k$.}
\item This process terminates when every student has been either assigned a tentative seat at a school or rejected by every school on her ROL, in which case she remains unassigned. At this point, all tentative acceptances become permanent.
\end{itemize}
\end{definition}

We say that a mechanism is \emph{strategyproof} if truthful reporting is a dominant strategy for every student. BM is not strategyproof: a student may benefit from reporting an ROL that differs from her true preference ordering. DA is strategyproof \citep{dubins1981machiavelli,roth1982economics}: regardless of other students' reported ROLs, it is a dominant strategy for every student to report her true preference ordering as her ROL.

\subsection{Multi-District School Choice}

In this paper, we extend the traditional school choice problem (henceforth, the \emph{single-district school choice problem}) by considering multiple districts. Specifically, in a \emph{multi-district school choice problem}, there is a set of $\ell$ \emph{districts}. Each student $i$ resides in some district $d(i)$ and each school $s$ is located in some district $d(s)$. A student's preference ordering may be over schools in multiple districts (including districts in which the student does not reside), and a school's priority ordering is over all students across all districts. A student's ROL may only contain schools from a single district, however; we call this the district in which she \emph{enrolls}. Intuitively, this models the real-world setting wherein a student can only enroll in a single school district in a given year.

Each district uses its own school choice mechanism to determine assignments of the students who enroll in the district to the schools that are located in the district. Different districts may use the same mechanism or different mechanisms; regardless, the school assignments for each district are combined to form the school assignments for the multi-district school choice problem as a whole.

\subsection{Student Sophistication Types and Constraint Types}

As in \citet{pathak2008leveling}, a student is either \emph{sincere} or \emph{sophisticated}; this is known as her \emph{sophistication type}. Once a sincere student $i$ determines that she will enroll in some district~$j$, she submits her preference ordering restricted to schools in $j$ (i.e., $P_i$ with any schools not in $j$ removed) as her ROL. In other words, a sincere student reports her true preferences over schools in the district she enrolls in. On the other hand, a sophisticated student strategizes when submitting her ROL over schools in the district she enrolls in. In addition to having a sophistication type, in the multi-district setting a student is also either \emph{constrained} or \emph{unconstrained}; we refer to this as her \emph{constraint type}. A constrained student~$i$ can only enroll in the district in which she resides, $d(i)$; while an unconstrained student can enroll in any (single) district. 

Combining these two attributes, we have four categories of students: \emph{sincere-constrained}, \emph{sincere-unconstrained}, \emph{sophisticated-constrained}, and \emph{sophisticated-unconstrained}. The behavior of each of these is largely intuitive, with a sincere-constrained student reporting her true preferences over schools in her district of residence; a sophisticated-constrained student strategically choosing an ROL over schools in her district of residence; and a sophisticated-unconstrained student strategically choosing both a district to enroll in and an ROL to submit over schools in that district. As in \citet{pathak2008leveling}, sophisticated students who enroll in a DA district always use their dominant strategy within that district, i.e., truthfully rank the schools in that district. If an unconstrained student is indifferent between districts to enroll in (since each student has a strict preference ordering over schools she finds acceptable, such indifference can only occur if this student remains unassigned regardless of where she enrolls), then we assume that this student enrolls in her district of residence.

For the majority of this paper, we consider a sincere-unconstrained student to be one who reports her true preferences over schools in whichever district she enrolls in, and strategically chooses in which district to enroll accordingly. This is not the only reasonable definition of sincere-unconstrained students, and in \cref{sec:heuristics} we ensure that our results hold for a large class of other reasonable definitions as well; see \cref{sec:sincere-unconstrained} for a more in-depth discussion of sincere-unconstrained students.

\subsection{Uniform \texorpdfstring{$(n; k)$}{(n; k)} model}

Throughout this paper, we use examples of specific multi-district school choice problems to demonstrate particular phenomena, some of which stand in contrast to the propositions in \citet{pathak2008leveling} that hold for the single-district setting. To analyze how frequently such phenomena occur, we consider large random two-district school choice problems inspired by the uniform models of \citet{babaioff2019playing}. In the \emph{uniform $(n; k)$ model}, there are $2$ districts labelled $L$ and $R$.\footnote{Later, we denote schools in $L$ as $\ell_1, \ell_2, ...$ and schools in $R$ as $r_1, r_2, ...$ for ease of reading.} Collectively, $L$ and $R$ contain $n$ students $I=\{i_1, ..., i_n\}$ and $n$ schools $S=\{s_1, s_2, ..., s_n\}$, each with unit capacity (i.e., $q_s = 1$ for all $s\in S$).

Each student is either sincere or sophisticated; is either constrained or unconstrained; and resides in either district $L$ or district $R$. There are thus eight categories of students: one for each possible \textit{sophistication type -- constraint type -- district of residence} combination. A student's category is drawn independently of all other students' categories, and there is a positive probability of a student's category being any of the eight possibilities. As such, there exists some \emph{category-probability lower bound} $p > 0$, such that for each category, the probability of an arbitrary student being in this category is at least $p$. Each student's preference ordering over schools (which may include schools in any district) is drawn uniformly at random from among all possible (strict) preference orderings of length $k$.\footnote{This is a special case of the procedure used to draw preference lists in \citet{immorlica2005marriage}.}
Each student's preference ordering is independent of all other students' preference orderings, and of all students' categories.

Each school independently has probability $\nicefrac{1}{2}$ of being located in district $L$ and $\nicefrac{1}{2}$ of being located in district $R$.\footnote{These probabilities need only be constant and nonzero (and sum to $1$) for our results to hold, but we set them equal to avoid clutter.}
Finally, each school has a complete (strict) priority ordering over all students, drawn uniformly at random from the set of all such possible orderings, and independently of everything else. Thus, for any school and any two students $i_a$ and $i_b$, the probability that $i_a$ has priority over $i_b$ at that school is $\nicefrac{1}{2}$.

%% file: 03-soph_prefer_soph.tex
\section{Sophistication Types}\label{sec:sophistication_types}

In this section, we show that it is possible for a sophisticated student to prefer that (i.e., strictly benefit if) some sincere student becomes sophisticated. In fact, we prove that such students are abundant in large random markets in which at least one district uses BM. This result stands in contrast to Proposition~3 of \citet{pathak2008leveling}, which is that in the single-district setting when BM is used, all sophisticated students weakly suffer if any sincere student becomes sophisticated. Finally, we prove that this phenomenon cannot occur when all districts use~DA.

\subsection{Example: A sophisticated student may prefer for a sincere student to become sophisticated}\label{sophistication-example}

We first provide an illustrative example. Suppose that there are two districts, $L$ and $R$, with schools $\ell_1 \in L$ and $r_1, r_2, r_3, r_4 \in R$ respectively, where each school has unit capacity. District $L$ uses an arbitrary mechanism, while district $R$ uses BM. Further suppose that there are four students, $i_1$, $i_2$, $i_3$, and~$i_4$.

The students' sophistication types and preference orderings, as well as the schools' priority orderings, are shown in the table below.\footnote{The notation $s_a \succ s_b$ indicates a student preference for school $s_a$ over school $s_b$. The notation $i_a - i_b$ indicates a school priority for student $i_a$ over student $i_b$. Technically, a school's priority ordering must include all students. Here, for each school, we list only the priority ordering over students who find the school acceptable, as no other student would ever apply to the school. Other students could be placed anywhere in each school's priority ordering without affecting our results.} Students whose preference orderings contain schools in only one district reside in that district and have arbitrary constraint types. Students whose preference orderings contains schools in both districts are unconstrained and reside in an arbitrary district.

\noindent\begin{minipage}[t]{.5\linewidth}
\vspace{0.5pt}
\hfil Student preferences \hfil
\hrule
\begin{align*}
    \text{(sincere) } i_1: & r_2 \succ r_1 \succ \ell_1 \\
    \text{(sophisticated) } i_2: & \ell_1 \succ r_3 \\
    \text{(sincere) } i_3: & r_2 \\
    \text{(sincere) } i_4: & r_1 \succ r_4
\end{align*}
\end{minipage}%
\begin{minipage}[t]{.5\linewidth}
\vspace{0.5pt}
\hfil School priority orderings \hfil
\hrule
\begin{align*}
    \ell_1: & i_1 - i_2 \\
    r_1: & i_1 - i_4 \\
    r_2: & i_3 - i_1 \\
    r_3: & i_2 \\
    r_4: & i_4
\end{align*}
\vspace{0.5pt}
\end{minipage}

We show that $i_2$ strictly prefers for $i_1$ to become sophisticated. First, consider the original setting where $i_1$ is sincere-unconstrained, and recall that $i_1$ strategizes over districts but always reports a truthful ROL. Observe that $i_1$ enrolls in district $L$, as doing so matches her with school $\ell_1$ while enrolling in district $R$ results in her being unassigned. This is because if $i_1$ were to enroll in $R$, she would not be matched in the first round of BM, during which both $r_1$ and $r_2$ would be filled.

Since $i_1$ enrolls in district $L$, sophisticated-unconstrained student $i_2$ enrolls in district $R$ (and ranks only $r_3$) to avoid being unassigned. The matching process results in $i_1$ assigned to~$\ell_1$, $i_2$ assigned to~$r_3$, $i_3$ assigned to~$r_2$, and $i_4$ assigned to~$r_1$. This is the unique Nash equilibrium outcome. Note that in this outcome, $i_2$ is assigned to her second-choice school.

Suppose instead that $i_1$ becomes sophisticated. Student $i_1$ has no chance of being assigned to her first-choice school, and is guaranteed admittance to $r_1$ if and only if she enrolls in district $R$ and ranks $r_1$ first, so she does so. Student $i_2$ therefore chooses to enroll in district $L$ (and ranks only~$\ell_1$), as this (and only this) guarantees her admittance at $\ell_1$. The matching process results in $i_1$ assigned to~$r_1$, $i_2$ assigned to~$\ell_1$, $i_3$ assigned to~$r_2$, and $i_4$ is assigned to~$r_4$. This is the unique Nash equilibrium outcome. In this outcome, $i_2$ is assigned to her first-choice school, which is a strict improvement for her compared to when $i_1$ is sincere.

\subsection{Large-Market Analysis}

We generalize the example from Section~\ref{sophistication-example} to the uniform $(n; 3)$ model. (The same analysis also works in the uniform $(n; k)$ model for any constant $k\ge3$.)
We say that a sophisticated student $i_a$ strictly (weakly) prefers for a sincere student $i_b$ to become sophisticated if $i_a$ strictly (weakly) prefers her match in every Nash equilibrium of the multi-district school choice problem when $i_b$ is sophisticated to her match in every Nash equilibrium of the multi-district school choice problem when $i_b$ is sincere.
We show that there can be many sophisticated students who strictly prefer for distinct sincere students to become sophisticated.

\begin{theorem}\label{thm:sophisticated}
    For every $p \in (0, 1)$, there exists $\tau > 0$ such that for any large enough $n$, in the uniform $(n; 3)$ model with category-probability lower bound $p$ and with one district using BM and the other using an arbitrary mechanism, there exists a set of sophisticated students of expected size at least $\tau n$ wherein each sophisticated student strictly prefers for a distinct sincere student to become sophisticated, and weakly prefers for all other sincere students to become sophisticated.
\end{theorem}

\begin{proof}
    Without loss of generality, let $R$ be the district that uses BM and let $L$ be the other district. We start by lower bounding the expected number of ordered quartets of students $(i_1, i_2, i_3, i_4)$ that satisfy the following conditions (as in the example from Section~\ref{sophistication-example}): 

    \begin{enumerate}
        \item \textbf{Conditions on sophistication types, constraint types, and residence:} 
        \begin{enumerate}
            \item $i_1$ is sincere and unconstrained.
            \item $i_2$ is sophisticated and unconstrained.
            \item Both $i_3$ and $i_4$ are sincere and reside in $R$.
        \end{enumerate}
        \item \textbf{Conditions on student preferences and school locations:}\footnote{These have been reordered from the example in Section~\ref{sophistication-example} to match their order in the analysis below.} 
        \begin{enumerate}
            \item $i_4$ most prefers a school $r_1 \in R$ and second most prefers a school $r_4\in R$.
            \item $i_3$ most prefers a school $r_2 \in R$.
            \item $i_2$ most prefers a school $\ell_1 \in L$ and second most prefers a school $r_3 \in R$.
            \item $i_1$ has preference ordering $r_2 \succ r_1 \succ \ell_1$.
            \item No student finds a school in $\{\ell_1, r_1, r_2, r_3,r_4\}$ acceptable except as above.\footnote{That is: $i_4$ does not find $\ell_1$, $r_2$, or $r_3$ acceptable; $i_3$ does not find $\ell_1$, $r_1$, $r_3$, or $r_4$ acceptable; and so forth; and, no student other than $i_1$, $i_2$, $i_3$, and $i_4$ finds $\ell_1, r_1, r_2, r_3$, or $r_4$ acceptable.}
        \end{enumerate}
        \item \textbf{Conditions on school priorities:} 
        \begin{enumerate}
            \item $i_1$ has priority over $i_2$ at $\ell_1$.
            \item $i_1$ has priority over $i_4$ at $r_1$.
            \item $i_3$ has priority over $i_1$ at $r_2$.
        \end{enumerate}
    \end{enumerate}
    
    As in the example in Section~\ref{sophistication-example}, under these conditions, $i_2$ strictly benefits from $i_1$ becoming sophisticated.
    
    We lower bound the probability that this set of conditions occurs for a specific $(i_1, i_2, i_3, i_4)$. Each of the three sets of conditions is independent. The conditions on sophistication types, constraint types, and residence are satisfied with probability at least $p^4$. The conditions on student preferences and school locations are satisfied with probability at least
    \begin{multline*}
        \left(\frac{1}{2} \cdot \frac{1}{2}\right) \cdot \left(\frac{n-2}{n}\cdot\frac{1}{2}\right)\cdot \left(\frac{n-3}{n} \cdot \frac{1}{2} \cdot \frac{n-4}{n-1}\cdot \frac{1}{2} \right) \cdot \\* \cdot \left(\frac{1}{n(n-1)(n-2)}\right) \cdot \left(\frac{n-5}{n} \cdot \frac{n-6}{n-1} \cdot \frac{n-7}{n-2}\right)^{n}.
    \end{multline*}
    In the expression above, each of the five parenthetical expressions corresponds to one of the five conditions on student preferences and school locations, and represents the condition's probability (for the last expression: a lower bound on the condition's probability), conditioned on the previous conditions. 
    For sufficiently large $n$,\footnote{Note that $\lim_{n\rightarrow\infty}(1-\frac{r}{n-a})^{(n-b)}$, where $a$, $b$, and $r$ are constants, is $e^{-r}$.} this probability can be lower bounded by\footnote{The final $\nicefrac{1}{2}$ is to accommodate for all the multiplicands of the form $\frac{n-a}{n-b}$ or $(1-\frac{r}{n-a})^{(n-b)}$, which have a nonzero limit, not fully converging for finite large $n$.}
    \[
        \frac{1}{2^5} \cdot \frac{(n - 3)!}{n!} \cdot \frac{1}{e^{15}} \cdot \frac{1}{2}.
    \]
    Finally, the conditions on school priorities are satisfied with probability $\nicefrac{1}{2^3}$. 
    
    The number of ordered quartets of distinct students is $\tfrac{n!}{(n - 4)!}$. List all such quartets and let the $j$th element of the list be $I_j$. Let $\mathbbm{1}_{I_j}$ be the indicator random variable that represents whether the students in $I_j$ satisfy all of the above conditions. Then the expected number of quartets of students that satisfy all conditions is 
    \[
        \mathbb{E}\left[\sum_{j =1 }^{\tfrac{n!}{(n - 4)!}}\mathbbm{1}_{I_j}\right] \geq \frac{n!}{(n - 4)!} \cdot p^4 \cdot \frac{1}{2^3} \cdot \frac{1}{2^{6} \cdot e^{15}} \cdot \frac{(n - 3)!}{n!} = \frac{p^{4}(n - 3)}{2^{9} \cdot e^{15}}.
    \]
    We observe that by construction, it is impossible for there to be two different quartets $(i_1,i_2,i_3,i_4)$ and $(i'_1,i'_2,i'_3,i'_4)$ that satisfy the above conditions such that $i'_1=i_1$ or $i'_2=i_2$. Assume for contradiction that for some $a\in\{1,2\}$, it is the case that $i'_a=i_a$ for two such quartets. If $a=2$ then we note that $i'_1=i_1$ as well, since $i_1$ is the only non-$i_2$ student who finds acceptable the school most preferred by $i_2$, and $i'_1$ is the only non-$i'_2$ student that finds acceptable the school most preferred by $i'_2=i_2$. Therefore, regardless of the value of $a$, we have that $i'_1=i_1$. However, this implies that $(i_1,i_2,i_3,i_4)=(i'_1,i'_2,i'_3,i'_4)$ because the identity of $i_1$ similarly uniquely determines those of $i_2$, $i_3$, and $i_4$, while the identity of $i'_1=i_1$ uniquely determines those of $i'_2$, $i'_3$, and $i'_4$ in the same way. This is a contradiction.

    To conclude, all such quartets of students involve distinct ``$i_1$'' and ``$i_2$'' students, so there are at least an expected $\frac{p^{4}}{2^{9} \cdot e^{15}} \cdot (n - 3)$ sophisticated students who each strictly prefer for a distinct sincere student to become sophisticated. 
    Also by construction, each such sophisticated student is unaffected by sincere students outside of her quartet, and the sophistication types of students $i_3$ and $i_4$ in the quartet do not affect the dynamic between students $i_2$ and~$i_1$.\footnote{This is because $i_3$ will always be assigned to her first choice school, and $i_4$ will either be assigned to $r_1$ if $i_1$ is sincere or $r_4$ if $i_1$ is sophisticated.} Thus, these sophisticated students also weakly prefer for all other sincere students to become sophisticated.    
    We choose $\tau = \frac{p^{4}}{2^{10} \cdot e^{15}}$ to satisfy the theorem statement.\footnote{We made no attempt to optimize this value, as our goal is only to ascertain linearity in $n$ of the number of such students.}
\end{proof}

\subsection{Necessity of a BM District}

We conclude this section by proving that the condition in Theorem \ref{thm:sophisticated} of at least one district using BM is not an artifact of our proof, but rather dropping this condition invalidates the result.

\begin{proposition}
For every multi-district school choice problem in which all districts use DA, there does not exist a sophisticated student and sincere student pair such that the sophisticated student strictly prefers for the sincere student to become sophisticated.
\end{proposition}

\begin{proof}
Let there be a multi-district school choice problem with any number of districts, all of which use DA. Consider a sincere student $i$. We first show that $i$ enrolls in the same district whether she is sincere or sophisticated. This is straightforward if $i$ is constrained, as then she always enrolls in her district of residence $d(i)$.

Otherwise, suppose that $i$ is unconstrained. Let $d$ be the district that $i$ enrolls in when sincere. By the definition of sincere-unconstrained, $i$ weakly prefers enrolling in $d$ and truthfully ranking schools in $d$ to doing the same in any other district. 

Note that because every district uses DA, if $i$ becomes sophisticated and enrolls in some district $d'$, she still truthfully ranks the schools in $d'$, thus obtaining the same outcome as if she enrolled in $d'$ when sincere. Therefore, if $i$ becomes sophisticated, she again weakly prefers enrolling in $d$. If $i$ is assigned to any school when she enrolls in $d$ as a sincere student, then this preference is strict because $i$ has a strict preference ordering over schools. If $i$ is unassigned when she enrolls in $d$ as a sincere student, then she must be unassigned regardless of what district she enrolls in and what her sophistication level is, and $d=d(i)$. In either case, $i$ still enrolls in $d$ if she becomes sophisticated.

To conclude, $i$ behaves the same whether she is sincere or sophisticated: We showed that $i$ enrolls in the same district regardless of her sophistication level, and she reports her true preferences over schools in that district. As such, every student has the same outcome if $i$ becomes sophisticated. Since $i$ was arbitrary, it is therefore not possible for a sophisticated student to prefer for any sincere student to become sophisticated.
\end{proof}

%% file: 04-soph_can_prefer_da.tex
\section{Mechanism Choice}\label{sec:mechanism_choice}

In this section, we show that another key result of \citet{pathak2008leveling} no longer holds true in the multi-district setting. Specifically, we show that in the multi-district setting, a sophisticated student may strictly prefer for her district to use DA instead of BM. This is true regardless of whether the sophisticated student is constrained or unconstrained. In particular, we give an example where the sophisticated student only finds schools in one district acceptable and prefers for that district to use DA.

\subsection{Example: A sophisticated student may strictly prefer DA}\label{soph_can_prefer_da-example}

Our example is as follows. Suppose there are two districts, $L$ and $R$, with schools $\ell_1, \ell_2 \in L$ and school $r_1 \in R$ respectively, where each school has unit capacity. Further suppose that there are three students, $i_1$, $i_2$, and $i_3$.

The students' sophistication types and preference orderings, as well as the schools' priority orderings, are shown in the table below.
Students $i_1$ and $i_2$ reside in $L$ and have arbitrary constraint types. Student $i_3$ is unconstrained and resides in an arbitrary district.

\noindent\begin{minipage}[t]{.5\linewidth}
\vspace{0.5pt}
\hfil Student preferences \hfil
\hrule
\begin{align*}
    \text{(sincere) } i_1: & \ell_1 \succ \ell_2 \\
    \text{(sophisticated) } i_2: & \ell_2 \succ \ell_1 \\
    \text{(sophisticated) } i_3: & \ell_2 \succ r_1
\end{align*}
\end{minipage}%
\begin{minipage}[t]{.5\linewidth}
\vspace{0.5pt}
\hfil School priority orderings \hfil
\hrule
\begin{align*}
    \ell_1: & i_2 - i_1 \\
    \ell_2: & i_1 - i_3 - i_2 \\
    r_1: & i_3
\end{align*}
\vspace{0.5pt}
\end{minipage}

Observe that because $i_3$ is the only student who finds any school in district $R$ acceptable, the mechanism used by district $R$ is irrelevant. 

We show that $i_2$ prefers for district $L$ to use DA rather than BM. First, assume that district~$L$ is using BM. As sincere student $i_1$ does not apply to $\ell_2$ in the first round, sophisticated student $i_3$ chooses to apply to $\ell_2$ in the first round, guaranteeing $i_3$'s acceptance at~$\ell_2$. Sophisticated student $i_2$ hence realizes that she has no chance at $\ell_2$ and instead applies to~$\ell_1$. This process results in $i_1$ unassigned, $i_2$ assigned to $\ell_1$, and $i_3$ assigned to $\ell_2$, which is the unique Nash equilibrium outcome. Note that in this Nash equilibrium outcome, $i_2$ is assigned to her second-choice school.

Suppose instead that district $L$ uses DA. Then,
student $i_2$ uses her dominant strategy of ranking $\ell_2$ above $\ell_1$. This induces student $i_3$ to enroll in $R$ instead of $L$, as enrolling in $L$ would result in $i_3$ being unassigned. Specifically, $i_3$ could ``knock out'' $i_2$ from~$\ell_2$, but in that case $i_2$ would knock out $i_1$ from $\ell_1$ and $i_1$ would in turn knock out $i_3$ from $\ell_2$. This process results in
$i_1$ assigned to $\ell_1$, $i_2$ assigned to $\ell_2$, and $i_3$ assigned to $r_1$, which is the unique Nash equilibrium outcome. In this Nash equilibrium outcome, $i_2$ is assigned to her first-choice school, which is a strict improvement for her compared to if district $L$ uses BM.

A key feature of this example is that there exists a cycle within the preferences of $i_1$ and~$i_2$. This causes $i_3$ to get knocked out of $\ell_2$ when $i_3$ applies to district $L$ and district $L$ is using DA. We show in Appendix~\ref{cycle-with-many-sophisticated} that this example can be generalized to include a cycle with additional sophisticated students, in which each sophisticated student within the cycle similarly strictly prefers for the district containing her entire preference list to use DA.

\subsection{Large-Market Analysis}

We generalize the example from Section~\ref{soph_can_prefer_da-example} to the uniform $(n; 2)$ model. (The same analysis also works in the uniform $(n; k)$ model for any constant $k\ge 2$.)
We say that a sophisticated student $i$ strictly prefers for a district $d$ to use DA if $i$ strictly prefers her match in every Nash equilibrium of the multi-district school choice problem when $d$ uses DA to her match in every Nash equilibrium of the multi-district school choice problem when $d$ uses BM.
We show that there can be a constant number of sophisticated students who each prefer for the district that contains her entire preference list to use DA.

\begin{theorem}\label{soph_can_prefer_da-theorem}
    For every $p \in (0, 1)$, there exists $\tau > 0$ such that for any large enough $n$, in the uniform $(n; 2)$ model with category-probability lower bound $p$, there exists a set of sophisticated students of expected size at least~$\tau$ wherein each sophisticated student strictly prefers for the district that contains her entire preference list to use DA rather than BM, regardless of the mechanism used by the other district.
\end{theorem}

\begin{proof}
    We start by lower bounding the expected numbered of ordered trios of students $(i_1, i_2, i_3)$ that satisfy the following conditions (as in the example from Section~\ref{soph_can_prefer_da-example}):
    \begin{enumerate}
        \item \textbf{Conditions on sophistication types, constraint types, and residence:} 
        \begin{enumerate}
            \item $i_1$ is sincere and resides in $L$.
            \item $i_2$ is sophisticated and resides in $L$.
            \item $i_3$ is sophisticated and unconstrained.
        \end{enumerate}
        \item \textbf{Conditions on student preferences and school locations:} 
        \begin{enumerate}
            \item $i_1$ most prefers a school $\ell_1 \in L$, and second most prefers a school $\ell_2 \in L$.
            \item $i_2$ has preference ordering $\ell_2 \succ \ell_1$.
            \item $i_3$ most prefers $\ell_2$ and second most prefers a school $r_1 \in R$.
            \item No student finds a school in $\{\ell_1, \ell_2, r_1\}$ acceptable except as above.
        \end{enumerate}
        \item \textbf{Conditions on school priorities:} 
        \begin{enumerate}
            \item $i_2$ has priority over $i_1$ at $\ell_1$.
            \item At $\ell_2, i_1$ has priority over $i_3$, who in turn has priority over $i_2$.
        \end{enumerate}    
    \end{enumerate}

    As in the example in Section~\ref{soph_can_prefer_da-example}, under these conditions, $i_2$ strictly benefits from $L$ using DA rather than BM.

    We lower bound the probability that this set of conditions occurs for a specific $(i_1, i_2, i_3)$. Each of the three sets of conditions is independent. The conditions on sophistication types, constraint types, and residence are satisfied with probability at least~$p^3$. The conditions on student preferences and school locations are satisfied with probability at least
    \[
        \left(\frac{1}{2}  \cdot \frac{1}{2} \right) \cdot \left(\frac{1}{n} \cdot \frac{1}{n-1}\right) \cdot \left(\frac{1}{n} \cdot \frac{n-2}{n-1} \cdot \frac{1}{2} \right)  \cdot \left(\frac{n-3}{n} \cdot \frac{n-4}{n-1}\right)^{n}.
    \]
    In the expression above, each of the four parenthetical expressions corresponds to one of the four conditions on student preferences and school locations, and represents a lower bound on the condition's probability, conditioned on the previous conditions. For sufficiently large~$n$, the above expression can be lower bounded by 
    \[
        \frac{1}{2^3} \cdot \frac{1}{n(n-1)(n-2)} \cdot \frac{1}{e^6}\cdot\frac{1}{2}.
    \]
    Finally, the conditions on school priorities are satisfied with probability $\tfrac{1}{2} \cdot \tfrac{1}{6} = \frac{1}{2^2 \cdot 3}$.
    
    The number of ordered trios of distinct students is $n(n-1)(n-2)$. List all such trios and let the $j$th element of the list be $I_j$. Let $\mathbbm{1}_{I_j}$ be the indicator random variable that represents whether the students in $I_j$ satisfy all of the above conditions. Then the expected number of trios of students that satisfy all conditions is
    \[
        \mathbb{E}\left[\sum_{j = 1}^{n(n-1)(n-2)} \mathbbm{1}_{I_j}\right] \geq n(n-1)(n-2) \cdot p^3 \cdot \frac{1}{2^2 \cdot 3} \cdot \frac{1}{2^4 \cdot e^6} \cdot \frac{1}{n(n-1)(n-2)} = \frac{p^3}{2^6 \cdot 3 \cdot e^{6}}.
    \]
    We observe that by construction, it is impossible for there to be two different trios $(i_1,i_2,i_3)$ and $(i'_1,i'_2,i'_3)$ that satisfy the above conditions such that $i'_2=i_2$. Indeed, this would imply that $(i_1,i_2,i_3)=(i'_1,i'_2,i'_3)$ because the identity of $i_2$ uniquely determines those of $i_1$ and $i_3$ (these are the only two other students who find acceptable the first choice of $i_2$, preferring it second and first respectively), while the identity of $i'_2=i_2$ uniquely determines those of $i'_1$ and $i'_3$ in the same way.

    To conclude, all such trios of students involve distinct ``$i_2$'' students, so there are at least an expected $\frac{p^3}{2^6 \cdot 3 \cdot e^{6}}$ sophisticated students who each strictly prefer for the district that contains their entire preference list to use DA over BM. We  choose $\tau = \frac{p^3}{2^6 \cdot 3 \cdot e^{6}}$ to satisfy the theorem statement.
\end{proof}

%% file: 05-constrained.tex
\section{Constraint Types}\label{sec:constraint_types}

In the previous two sections, we have shown that two predictions made by \citet{pathak2008leveling} no longer hold in multi-district settings. The first of these predictions revolved around whether a sophisticated student might prefer for another student to change their sophistication type in a particular way (specifically, becoming sophisticated). In our setting, students are not only characterized by their sophistication type but also by their constraint type, and hence it is natural to ask, for completeness, whether and when some students might prefer for some other students to change their constraint type.

We find a strong ``anything goes'' result here: For any combination of constraint types for two students, it might be the case that the first student strictly prefers for the constraint type of the second student to change, and this is furthermore abundant in large random markets. This holds regardless of the sophistication types of the two students, regardless of whether or not they reside in the same district, and regardless of the mechanisms used by the two districts. (We state this result in the uniform $(n; 2)$ model; however, similarly to our previous results, the same analysis also works in the uniform $(n; k)$ model for any constant $k\ge 2$.)
We say that a student $i_a$ strictly prefers for a student $i_b$ to change her constraint type if $i_a$ is strictly worse off in every Nash equilibrium of the multi-district school choice problem when $i_b$ has her given constraint type compared to every Nash equilibrium of the multi-district school choice problem when $i_b$ has the opposite constraint type.

\begin{theorem}\label{constraint}
    For every $p \in (0, 1)$, there exists $\tau > 0$ such that for every pair of (not necessarily distinct) sophistication types $s_1$ and $s_2$, every pair of (not necessarily distinct) constraint types $c_1$ and~$c_2$, and for any large enough $n$, in the uniform $(n; 2)$ model with category-probability lower bound $p$ and with the districts using any combination of matching mechanisms, there exists a set of expected size at least $\tau n$ of students of sophistication type $s_1$ and constraint type $c_1$ wherein each student strictly prefers for a distinct student of sophistication type $s_2$ and constraint type $c_2$ from the same district to change her constraint type. Furthermore, this also holds if ``from the same district'' is replaced with ``from the other district.''
\end{theorem}

As it turns out, two types of constructions suffice to cover all of the various combinations of constraint types, sophistication types, districts of residence, and mechanisms in \mbox{\cref{constraint}}. Let $i_1$ and $i_2$ be two students, where $i_1$ is the student who prefers for $i_2$ to change her constraint type. Both constructions involve $i_2$ vacating her seat at a school $s$---either because she becomes constrained and $s$ is not in her district of residence, or because she becomes unconstrained and would rather enroll in another district.

The simpler of the two constructions has $i_1$ filling the vacancy left by $i_2$ at $s$. This construction can be used in cases where $i_1$ is unconstrained, as well as in cases where $i_1$ is constrained and the school $s$ is in $i_1$'s district of residence. The second, slightly more elaborate construction covers the remaining cases, in which $i_1$ is constrained and the school~$s$ is outside $i_1$'s district of residence. These cases, in which our results are arguably more surprising at first glance, include situations where $i_2$ resides in the same district as (the constrained) $i_1$ and vacates a spot in another district when $i_2$ becomes constrained, as well as situations where $i_2$ resides in a different district than $i_1$ and vacates a spot in that district when $i_2$ becomes unconstrained. In such cases, we introduce a third, unconstrained student, $i_3$, who takes the seat vacated by $i_2$ and therefore vacates a seat in the district of $i_1$, which~$i_1$ in turn gets to fill.

In \cref{app:constraint}, we prove two cases of \cref{constraint}---one using each of the two constructions. The proof of each of the other cases of \cref{constraint} is completely analogous to the proof of one of these two cases.

%% file: 06-sincere_unconstrained.tex
\section{Sincere-Unconstrained Students}\label{sec:sincere-unconstrained}

In this section, we dive deeper into the machinery behind the results of \cref{sec:sophistication_types}, and ask which features of the school choice problems used in the analysis of that \lcnamecref{sec:sophistication_types} are necessary to give rise to the phenomenon of sophisticated students who prefer for sincere students to become sophisticated.

The definitions of three of the student types in our taxonomy follow closely from ideas in \citet{pathak2008leveling}: the two constrained types behave, within their district, just like the two types in \citet{pathak2008leveling}, while sophisticated-unconstrained students have complete strategic freedom, analogous to sophisticated students having complete strategic freedom within the framework of \citet{pathak2008leveling}. Our fourth type---sincere-unconstrained students---while still heavily inspired by \citet{pathak2008leveling} and completely in line with our taxonomy, is conceptually (and formally, as we discuss below) somewhat further from 
the types defined by that paper. Therefore, we first ask whether and to what extent this student type plays a key role in our analysis, and then ask whether and to what extent our results are robust to changes in the precise definition of such students.

\subsection{The Role of Sincere-Unconstrained Students in the Results of Section~\ref{sec:sophistication_types}}

\subsubsection{Example: When all sincere students are constrained}\label{sincere-unconstrained-example}

The construction in \cref{sec:sophistication_types} hinges on a trait that only a sincere-unconstrained student can possess: She can block a sophisticated student from getting a seat at some school, and yet, she may vacate that seat for another school if she becomes sophisticated.\footnote{While a sincere-constrained student can also block a sophisticated student in the same way, a sincere-constrained student will not vacate that seat if she becomes sophisticated, as it is the best outcome she can achieve.} The first question we ask is whether there exists a school choice problem in which a sophisticated student prefers for a sincere student to become sophisticated even in the absence of sincere-unconstrained students. We give an affirmative answer, whose construction is the most involved in this paper so far.

\begin{proposition}\label{sophisticated-no-sincere-unconstrained}
There exists a two-district school choice problem in which (1) all sincere students are constrained, and (2) there exists a sophisticated student who strictly prefers for a specific sincere student to become sophisticated.
\end{proposition}

\begin{proof}
Let there be two districts with schools $\ell_1,\ell_2,\ell_3 \in L$ and schools $r_1, r_2, r_3, r_4 \in R$, where each school has unit capacity. District $L$ uses BM, while district $R$ uses DA. Furthermore, let there be six students $i_1$, $i_2$, $i_3$, $i_4$, $i_5$, and~$i_6$.

The students' sophistication types and preference orderings, as well as the schools' priority orderings, are shown in the table below. Students whose preference orderings contain schools in only one district reside in that district and are constrained to it. Other students (those whose preference orderings contains schools in both districts) are unconstrained and reside in an arbitrary district.

\noindent\begin{minipage}[t]{.5\linewidth}
\vspace{0.5pt}
\hfil Student preferences \hfil
\hrule
\begin{align*}
    \text{(sincere) }i_1: & \ell_1 \succ \ell_2 \\
    \text{(sophisticated) }i_2: & \ell_2 \succ r_3 \\
    \text{(sophisticated) }i_3: & r_1 \succ r_2 \\
    \text{(sophisticated) }i_4: & r_2 \succ r_1 \\
    \text{(sophisticated) }i_5: & r_1 \succ r_3 \succ \ell_3 \\
    \text{(sincere) }i_6: & \ell_1
\end{align*}
\end{minipage}%
\begin{minipage}[t]{.5\linewidth}
\vspace{0.5pt}
\hfil School priority orderings \hfil
\hrule
\begin{align*}
    \ell_1: & i_6 - i_1 \\
    \ell_2: & i_1 - i_2 \\
    \ell_3: & i_5 \\
    r_1: & i_4 - i_5 - i_3\\
    r_2: & i_3 - i_4 \\
    r_3: & i_2 - i_5
\end{align*}
\vspace{0.5pt}
\end{minipage}

We will show that $i_3$ and $i_4$ prefer for $i_1$ to become sophisticated. First, consider the original setting where $i_1$ is sincere (and constrained). Since $i_6$ both ranks $\ell_1$ highest and has the highest priority at that school among all students who find it acceptable, $i_6$ is assigned to $\ell_1$. Observe that the sophisticated-unconstrained student $i_2$ enrolls in district $L$, since she is thus assigned to $\ell_2$ as the only student who applies to this school in the first round. Therefore, $i_1$ remains unassigned. Sophisticated-unconstrained student $i_5$ enrolls in district $R$ because this gets her assigned to school $r_3$, which she strictly prefers to any district $L$ school. The matching process results in $i_1$ unassigned, $i_2$ assigned to $\ell_2$, $i_3$ assigned to $r_2$, $i_4$ assigned to $r_1$, $i_5$ assigned to $r_3$, and $i_6$ assigned to $\ell_1$.

Suppose instead that $i_1$ becomes sophisticated. By the same argument, $i_6$ is assigned to $\ell_1$. Diverging from the previous analysis, $i_1$ applies to school $\ell_2$ in the first round and gets assigned to it. $i_2$ therefore enrolls in district $R$ rather than $L$, and since she has the highest priority at $r_3$ among all students who find it acceptable, $i_2$ is assigned to $r_3$. Therefore, since enrolling in district $R$ would now leave student $i_5$ unassigned, she enrolls in district $L$ and is assigned to $\ell_3$. The matching process results in $i_1$ assigned to $\ell_2$, $i_2$ assigned to $r_3$, $i_3$ assigned to $r_1$, $i_4$ assigned to $r_2$, $i_5$ assigned to $\ell_3$, and $i_6$ assigned to $\ell_1$. Therefore, both $i_3$ and $i_4$ strictly prefer for $i_1$ to become sophisticated.
\end{proof}

\subsubsection{Shortcomings of Large-Market Analysis without Sincere-Unconstrained Students}

We observe that replacing the gadget in our current proof of \cref{thm:sophisticated} by the school choice problem from \cref{sophisticated-no-sincere-unconstrained} results in a weaker guarantee. Specifically, the resulting guarantee is of an expected number of sophisticated students who prefer for some sincere student to become sophisticated that is only constant, rather than a constant fraction of the total number of students.\footnote{This is due to the existence of a cycle in the preferences of $i_4$ and $i_5$, similar to the one in the preferences of $i_1$ and~$i_2$ in the proof of \cref{soph_can_prefer_da-theorem}. We thus get a quantitative guarantee similar to that of \cref{soph_can_prefer_da-theorem} rather than to that of \cref{thm:sophisticated}.}

Consider a school choice problem in which some sophisticated student strictly prefers for a sincere student to become sophisticated. We call this school choice problem \emph{weak} if using it to replace the gadget in the proof of \cref{thm:sophisticated} results in a constant-number or weaker guarantee. We call this school choice problem  \emph{strong} if using it to replace the gadget in the proof of that theorem results in the constant-fraction guarantee as in the statement of that theorem. In particular, the school choice problem from \cref{sophistication-example} is strong and the school choice problem from the proof of \cref{sophisticated-no-sincere-unconstrained} is weak.

We now argue that being weak is a limitation not only of the school choice problem from the proof of \cref{sophisticated-no-sincere-unconstrained}, but rather of every school choice problem that lacks sincere-unconstrained students. We start by investigating which aspects of the machinery of \citet{pathak2008leveling} transfers to our setting, in what way, and with what limitations.

Fix some multi-district school choice problem (possibly containing sincere-unconstrained students). Without loss of generality, assume that each constrained student only has in its preference list schools that are in its district. Generalizing a construction by \citet{pathak2008leveling}, define the following \emph{single}-district school choice problem, to which we refer as the ``augmented economy.'' The schools and students in the augmented economy are identical to those in the original, multi-district school choice problem, but the priorities of each school originally located in a district that uses BM is modified as follows to partition students into tiers, where within each tier the priorities are the same as in the original school choice problem, but each student in the first tier has priority over each student in the second tier, and so on. For a school $s$, the first tier consists of all sophisticated students, as well as all sincere students who prefer $s$ over all other schools in the same district as $s$ in the original multi-district school choice problem.\footnote{Recall that students who are originally constrained to districts other than the school's do not have this school in the preference list, so while such sophisticated students are included in this tier, this is inconsequential.} For every $j\ge2$, the $j$th tier consists of all sincere students who strictly prefer precisely $j-1$ schools from the original district of $s$ over $s$.

The proof of Proposition~4 of \citet{pathak2008leveling} (which precludes the possibility of a result like our \cref{thm:sophisticated,sophisticated-no-sincere-unconstrained} in their model) hinges on two properties of the augmented economy in that paper: First, the set of equilibrium outcomes of the original school choice problem coincides with the set of stable matchings of the augmented economy (this is their Proposition 1); and second, every school that participates in BM gives priority in the augmented economy to all sophisticated students, as well as all sincere students who prefer the school most, over all other students (this is a condition of their Lemma~1). The latter property is violated only by sincere-unconstrained students in our augmented economy for multi-district school choice problems.\footnote{Indeed, this property is precluded by the trait discussed in the beginning of \cref{sincere-unconstrained-example} as driving all of the constructions from \cref{sec:sophistication_types}: that a sincere-unconstrained student can block a sophisticated student at some school $s$ (and hence have top-tier priority at $s$), and yet, this sincere-unconstrained student may vacate that seat for another school if she becomes sophisticated (and hence, when sincere, does not have top-tier priority at a school she prefers over $s$).} When all sincere students are constrained, this property is satisfied; however, the first property might be violated, as we now discuss.\looseness=-1

An equilibrium in our original multi-district school choice problem need not be stable in the augmented economy: Such an equilibrium might contain blocking pairs with respect to the augmented economy,\footnote{A blocking pair is a student--school pair where the student prefers the school to her current assignment, and the student has priority (in our case, in the augmented economy) at the school over another student currently assigned to the school.} but only with a very specific structure, to which we refer as \emph{incentive-compatible (IC) blocking pairs}. An IC blocking pair in an equilibrium of our original multi-district school choice problem is a student--school pair that is a blocking pair with respect to the augmented economy; and involves a school $s$ in a district that uses DA, and a (sophisticated) unconstrained student $i$ who enrolls (in the equilibrium) in a different district from $d(s)$, such that if $i$ were to deviate from the equilibrium and enroll in $d(s)$, then in the run of DA in $d(s)$, after $i$ applies to $s$ (knocking out the equilibrium match of $s$), a rejection cycle is initiated that eventually causes $i$ to be rejected from $s$. An example of an IC blocking pair can be seen in the school choice problem in the proof of \cref{sophisticated-no-sincere-unconstrained} when $i_1$ is sophisticated: $i_5$ (who in equilibrium enrolls in district $L$) blocks (in the augmented economy) with school $r_1$ (which is in district $R$, which uses DA). However, if $i_5$ were to enroll in district $R$, then in the run of DA in that district, after $i_5$ knocks out $i_3$ from $r_1$, a rejection cycle is initiated that eventually causes the rejection of $i_5$ from $r_1$ (in favor of $i_4$).

In the absence of both sincere-unconstrained students and IC blocking pairs, the proof of \citet{pathak2008leveling} in fact goes through in our setting as well. In a precise sense, therefore, sincere unconstrained students and IC blocking pairs are the \emph{only} two phenomena that might preclude the result of \citet{pathak2008leveling} in our setting; in \cref{sec:sophistication_types} we demonstrated how the former might do so, and in \cref{sophisticated-no-sincere-unconstrained} we demonstrated this for the latter. Therefore, any school choice problem that proves \cref{sophisticated-no-sincere-unconstrained} despite not containing sincere-unconstrained students must satisfy that its equilibrium matching exhibits an IC blocking pair.\footnote{Furthermore, such a school choice problem must have, in addition to a district that uses DA as mandated by the definition of an IC blocking pair, a district that uses BM for the sincere-turned-sophisticated student to reside in, because if this student resides in a district that uses DA, her behavior is unchanged when becoming sophisticated.} In particular, the preferences of the students must allow for a rejection cycle, which in turn means that there must exist a cycle of students such that for every two consecutive students in the cycle, there must exist a school that both students have in their preference orderings. We now argue that such a cycle of students necessarily implies the school choice problem is weak.

Fix some multi-district school choice problem that satisfies \cref{sophisticated-no-sincere-unconstrained}. Consider a hypergraph $H=(V,E)$ where $V$ are the students in the school choice problem and the hyperedges are in one-to-one correspondence to schools, where the hyperedge corresponding to a specific school contains all of the vertices (students) that have this school in their preference lists. Note that students in distinct connected components of $H$ cannot affect each other, so it's enough to consider each connected component that satisfies \cref{sophisticated-no-sincere-unconstrained} separately and prove that its induced sub-problem is weak (connected components without sophisticated students who strictly prefer for some sincere students to become sophisticated can be ignored). We therefore assume henceforth that $H$ is connected.

In the probabilistic analysis in the proof of \cref{thm:sophisticated} (as well as in the proof of our other probabilistic results), the order of magnitude of our bound on the expected number of occurrences of the gadget is completely determined by the corresponding hypergraph~$H$: Starting with a constant, for every vertex $v\in V$ our bound on the expected number of occurrences gains a factor of $n$, and for every hyperedge $e\in E$ this bound loses a factor of $n^{|e|-1}$. Therefore, if $|V|\le\sum_{e\in E}\bigl(|e|-1\bigr)$, then an expected number of occurrences that is a constant fraction of $n$ cannot be guaranteed using this proof, implying that the school choice problem that we fixed is weak. Note that the cycle of students that is guaranteed above to exist is in fact a cycle in $H$. As such, $H$ satisfies $\sum_{e\in E}\bigl(|e|-1\bigr)\ge|V|$ (recall that for every connected hypergraph, $\sum_{e\in E}\bigl(|e|-1\bigr)\ge|V|-1$, with equality only for hypertrees, i.e., for acyclic connected hypergraphs). Therefore, the school choice problem that we fixed is weak.

Overall, we have shown that no school choice problem in which all sincere students are constrained can be used as the gadget in our proof of \cref{thm:sophisticated}. While this technically does not rule out other proof techniques, it is a strong indication of a key role that sincere-unconstrained students play in the phenomenon formulated by \cref{thm:sophisticated}. Given this result, in the next section we ask how robust \cref{thm:sophisticated} is to the precise definition of sincere-unconstrained students.

\subsection{Robustness to the Definition of Sincere-Unconstrained Students}\label{sec:heuristics}

The definition of sincere-unconstrained students is not beyond criticism. Indeed, one might wonder why a student who possesses sufficient understanding of the intricacies of the mechanisms used by the different districts (and the preferences of other students) so as to be able to perfectly optimize between them does not realize the potential benefits of reporting strategically within the district in which she chooses to enroll. Having ascertained the key role of sincere-unconstrained students in the proof of \cref{thm:sophisticated}, it is therefore natural to ask how robust \cref{thm:sophisticated} is to weakening the ability of such students to fully optimize between districts. Specifically, sincere-unconstrained students might be more realistically modeled as students who are able (e.g., financially) to enroll in a district of their choice, but do not fully grasp the details of the mechanism used in each district and the preferences of other students (and hence report truthfully within the district that they end up choosing). Such students, therefore, are unable to fully optimize their choice of district and have to resort to using some heuristic to choose between districts. In this section, we consider sincere-unconstrained students who choose between districts based on a natural broad class of heuristics that we term \emph{positional heuristics}. We show that regardless of the specific heuristics within this broad class that are used, the main results of the previous sections still hold.

There are many reasonable heuristics that sincere-unconstrained students might employ to choose between districts, and different heuristics may lead to different choices of districts to enroll in. For example, a sincere-unconstrained student could always enroll in the district of her first-choice school (and then report her true preferences over schools in that district). This can, of course, be suboptimal: Such a sincere-unconstrained student who can never get assigned to her first-choice school will still enroll in the district of that school; in contrast, a sincere-unconstrained student who fully strategizes over her district choice (like those considered in the previous section) would instead enroll in another district if she would get a preferable school assignment by reporting her true preferences over schools in that district. 

Beyond fully strategizing over district choice and always enrolling in the district of her first-choice school, there are many other reasonable ways a sincere-unconstrained student can choose where to enroll. For example, a sincere-unconstrained student might simply count the number of schools on her preference ordering from each district, and enroll in the district with the most such schools. Interpolating between this and considering only her first-choice school, she could enroll in the district that has the most schools in, e.g., the top ten schools of her preference ordering. Alternatively, she might score each school based on how high up it is in her preference ordering, e.g., by giving her least preferred (but still acceptable) school one point, her next least preferred school two points, and so on, and then enrolling in the district with the highest sum of points over all schools.

We believe that all of these definitions have merit in different contexts, and therefore avoid choosing between them. Instead, we show a version of our results that holds for a large class of heuristics for sincere-unconstrained students that include all of the above and many more. Taking inspiration from social choice theory \citep[e.g.,][]{young1975social, boutilier2012optimal}, we define the class of scoring functions as follows.

\begin{definition}\leavevmode
\begin{itemize}
\item
    For $v \in \mathbb{R}^n$, define the \emph{scoring function} $f^v$ as follows. For any preference ordering $P_i$, define $P_i(s)$ as the ranking of school $s$ in $P_i$. Further let $v\bigl(P_i(s)\bigr) = v_{P_i(s)}$ if $s\in P_i$ and $v\bigl(P_i(s)\bigr)=0$ otherwise. Then we define the scoring function as
    \[
        f^v(P_i) = \argmax_d \smashoperator[r]{\sum_{s:d(s) = d}} v\bigl(P_i(s)\bigr).
    \]
\item
A scoring function is \emph{monotone} if both $v\bigl(P_i(s)\bigr) \geq v\bigl(P_i(s')\bigr)$ whenever $s$ is ranked higher than~$s'$ in~$P_i$ and $v\bigl(P_i(s)\bigr)\ge0$ whenever $s\in P_i$. 
\item 
A sincere-unconstrained student uses a \emph{positional heuristic} if she chooses in which district to enroll based on a monotone scoring function. In the case of a tie between districts, we assume that a student using a positional heuristic enrolls in the district of the school that she ranks first among all schools in the tied districts.
\end{itemize}
\end{definition}

We emphasize that all of the heuristics discussed above are positional heuristics. For example, one positional heuristic that corresponds to a sincere-unconstrained student choosing to enroll in the district of her first-choice school is the heuristic that assigns a score of $1$ to the first-choice school and a score of $0$ to every other school. If a sincere-unconstrained student instead enrolls in the district with the most schools they find acceptable, this is equivalent to using a positional heuristic that assigns a score of $1$ to every school in her preference ordering (and $0$ to every other school). If a student scores each school by how far it is from the end of her preference ordering, this is equivalent to using the positional heuristic that assigns a score of $[\text{length of ROL}+1 - \ell]$ to the school that is ranked $\ell$ on her list and is analogous to the well-known \emph{Borda Count} voting rule \citep{black1958theory}. 

As we show in this section, our main results also hold when each sincere-unconstrained student uses some positional heuristic to choose between districts (with different sincere-unconstrained students possibly using different heuristics or even fully optimizing).
We note that our results in Sections \ref{sec:mechanism_choice} and \ref{sec:constraint_types} do not require sincere-unconstrained students at all. Therefore, the results in those sections hold regardless of how a sincere-unconstrained student chooses between districts (even if a sincere-unconstrained student does not use a monotone scoring function). We therefore need only justify how our example and theorem from \cref{sec:sophistication_types} continue to hold when sincere-unconstrained students use positional heuristics. 

We first adapt the example from \cref{sophistication-example}. Recall that there is only one sincere-unconstrained student in this example, whose preferences are $r_2 \succ r_1 \succ \ell_1$, and who, when fully optimizing between districts but reporting truthfully within a district, enrolls in district $L$ (and causes $i_2$ to not get assigned to school $\ell_1$). Our key insight is that by adding two schools in district $L$ to the top of the preference list of this student, we can guarantee that she will always choose to enroll in district $L$, regardless of the monotone scoring function used (and consistent with full optimization). This is because by monotonicity, the score of each school in district $R$ is at most the score of the school ranked two above it in district $L$, so the overall score for district $R$ must be at most the score for district $L$. In order to ensure that $i_1$ is not assigned to these new district-$L$ schools when her type is sincere-unconstrained, we also need to add two students who reside in district $L$, each having one of the new schools as her top choice (and each having top priority at that respective school). The outcomes of this modified example match the outcome in Example \ref{sophistication-example}, except that the two new students are also assigned to the two new schools. To ensure that $i_1$, when enrolling in district $L$, still causes $i_2$ to not get assigned to school $\ell_1$, we have district $L$ use DA. The full example is given in \cref{app:heuristics}.

Finally, we adapt our large-market analysis to show that a version of Theorem \ref{thm:sophisticated} still holds when sincere-unconstrained students use positional heuristics, which we prove in \cref{app:heuristics}.

\begin{theorem}\label{thm:positional_heuristics}
    Suppose that every sincere-unconstrained student uses a positional heuristic to decide which district to enroll in. Then for every $p \in (0, 1)$, there exists $\tau > 0$ such that for any large enough $n$, in the uniform $(n; 5)$ model with category-probability lower bound $p$ and with one district using BM and the other using DA, there exists a set of sophisticated students of expected size at least $\tau n$ where each sophisticated student strictly prefers for a distinct sincere student to become sophisticated, and weakly prefers for all other sincere students to become sophisticated.
\end{theorem}

%% file: 07-discussion.tex
\section{Discussion}\label{sec:discussion}

In this paper, we show that several key results regarding sincere and sophisticated students from the seminal paper of \citet{pathak2008leveling} no longer hold in a multi-district setting. This highlights that when weighing the specifics of the market in question when designing centralized mechanisms, it might be useful to define ``the market in question'' more broadly than is customary.

Several aspects of our model and results would potentially benefit from further research. First, although the idea of district choice is motivated by students and their families boundary hopping or paying a premium to move, we do not explicitly model the associated (financial or risk-taking) costs. Future research could impose a price on choosing a district other than one’s own district of residence, which would then factor into students' strategic considerations. For this to be effective, it would also be necessary to think about student satisfaction with different school assignments using cardinal utilities rather than ordinal preferences; modeling such costs is therefore outside the scope of this paper.\footnote{A recent paper by \cite{ArtemovT2025} studies a single district in which students  can have ``walking distance'' priority for a school based on where they reside. Future research into multi-district school choice might draw inspiration from the way the model of \citeauthor{ArtemovT2025} incorporates a moving cost.}

While Theorem~\ref{soph_can_prefer_da-theorem} establishes an at least constant frequency of sophisticated students who strictly prefer DA over BM even in a large random market, it is our only theorem that does not prove the abundance of such students (i.e., an expected constant fraction of all sophisticated students in a large random market). Even if we extend the proof of Theorem~\ref{soph_can_prefer_da-theorem} to also take into account all cycles of the form discussed in Appendix~\ref{cycle-with-many-sophisticated}, the same proof technique would yield only constant frequency. To achieve a linear, or even super-constant frequency, one would have to, for example, find a way for the existence of some cycle to be sufficient for a super-constant number of sophisticated students outside the cycle to strictly prefer BM over DA. While we conjecture that this is not possible, ruling this out seems to be related to the question of whether Deferred Acceptance ``circuits'' can efficiently encode computational circuits in which various wires split, a question that was resolved negatively by \citet{cook2014complexity}. It is plausible that computational-complexity-theoretic tools such as those used in that paper might be leveraged to prove the asymptotic tightness of the bound in Theorem~\ref{soph_can_prefer_da-theorem}. We conjecture this bound to be tight (perhaps up to logarithmic factors if preference list lengths are not held constant), but leave the verification of this conjecture for future work.

%% file: 08-appendix.tex
\section{Cycles where Multiple Sophisticated Students Prefer Deferred Acceptance}\label{cycle-with-many-sophisticated}

In this appendix, we generalize the example in Section~\ref{soph_can_prefer_da-example} to show that there may be multiple sophisticated students who all prefer that a district uses DA when their preferences form a cycle including exactly one sincere student. 

Suppose again there are two districts $L$ and $R$ with schools $\ell_1, \ell_2,..., \ell_x \in L$ and school $r_1 \in R$, where all schools have capacity $1$. Further suppose that there are $x + 1$ students $i_1, i_2,...,i_{x+1}$. Student $i_1$ is sincere and constrained, while students $i_2,...,i_{x+1}$ are all sophisticated and unconstrained. 

\noindent\begin{minipage}[t]{.5\linewidth}
\vspace{0.5pt}
\hfil Student preferences \hfil
\hrule
\begin{align*}
    i_1: & \ell_1 \succ \ell_x \\
    i_2: & \ell_2 \succ \ell_1 \\
    & \vdots \\
    i_x: & \ell_x \succ \ell_{x-1} \\
    i_{x+1}: & \ell_x \succ r_1 \\
\end{align*}
\end{minipage}%
\begin{minipage}[t]{.5\linewidth}
\vspace{0.5pt}
\hfil School priority orderings \hfil
\hrule
\begin{align*}
    \ell_1: & i_2 - i_1 \\
    \ell_2: & i_3 - i_2 \\
    & \vdots \\
    \ell_{x-1}: & i_x - i_{x-1} \\
    \ell_x: & i_1 - i_{x+1} - i_x \\
    r_1: & i_{x+1} \\
\end{align*}
\vspace{0.5pt}
\end{minipage}

Note that the preferences of students $i_1,...,i_x$ are cyclical. Student $i_{x + 1}$ most prefers the second most preferred school of sincere student $i_1$, and second most prefers the only school in district $R$. Student $i_{x + 1}$ is also the only student who finds $r_1$ acceptable, which implies that the matching mechanism used by $R$ is irrelevant.

We will show that students $i_2,...,i_x$ all prefer district $L$ to use DA instead of BM. First, assume that district $L$ is using BM. As sincere student $i_1$ will not apply to $\ell_x$ in the first round, sophisticated student $i_{x + 1}$ will choose to apply to his first choice $\ell_x$ in the first round, guaranteeing acceptance at $\ell_x$ for himself. Sophisticated student $i_x$ will realize she has no chance at $\ell_x$ and will instead apply to $\ell_{x - 1}$. This initiates a chain reaction in which all students $i_2,...,i_{x}$ end up applying to their second choice school, as each would not be accepted by their first choice school. This process results in $i_1$ unassigned, $i_2,...,i_x$ each assigned to their second choice schools $\ell_1,...,i_{x-1}$ respectively, and $i_{x + 1}$ assigned to $\ell_x$. 

Now suppose instead that district $L$ uses DA. Then each of the students $i_2,...,i_x$ uses her dominant strategy of ranking truthfully. This induces student $i_{x + 1}$ to enroll in $R$ instead of $L$, as enrolling in $L$ would result in $i_{x+1}$ being unassigned. Note that if $i_{x+1}$ chooses to enroll in $L$, $i_{x+1}$ would end up being knocked out of $\ell_x$ by $i_1$, and would therefore be unassigned as $i_{x+1}$ finds no other schools in $L$ acceptable. Therefore, the unique Nash equilibrium has $i_1,...,i_x$ each assigned to their first choice schools $\ell_1,...,i_x$ respectively, and $i_{x + 1}$ assigned to $r_1$. In this Nash equilibrium, sophisticated students $i_2,...,i_x$ are all assigned to their first choice schools, which is a strict improvement for each compared to when district $L$ uses BM.

\section{Proofs Omitted from Section~\ref{sec:constraint_types}}\label{app:constraint}

We start by demonstrating the first, simpler construction for proof of a case of \cref{constraint}.

\begin{lemma}
    For every $p \in (0, 1)$, there exists $\tau > 0$ such that for any large enough $n$, in the uniform $(n; 2)$ model with category-probability lower bound $p$ and with the districts using any combination of matching mechanisms, there exists a set of unconstrained students of expected size at least $\tau n$ where each unconstrained student strictly prefers for a distinct constrained student to become unconstrained.
\end{lemma}

\begin{proof}
    We start by lower bounding the expected number of ordered pairs of students $(i_1, i_2)$ that satisfy the following conditions:
    \begin{enumerate}
        \item $i_1$ is unconstrained, and $i_2$ is constrained and resides in district $L$.
        \item $i_1$ most prefers a school $\ell \in L$.
        \item $i_2$ most prefers a school $r \in R$, and second most prefers $\ell$.
        \item No student finds a school in $\{\ell, r\}$ acceptable except as above.. 
        \item $i_2$ is higher in the priority order than $i_1$ at $\ell$.
    \end{enumerate}
    Under these conditions, $i_1$ strictly benefits from $i_2$ becoming unconstrained, as $i_2$ then switches from enrolling in district $L$ to enrolling in district $R$, which opens up school~$\ell$ for $i_1$.
    
    We lower bound the probability that this set of conditions occurs for a specific $i_1$ and~$i_2$. Conditioned on all previous conditions, the first condition occurs with probability at least $p^2$, the second condition occurs with probability at least $\tfrac{1}{2}$, and the third condition occurs with probability $\tfrac{n-1}{n} \cdot \tfrac{1}{2} \cdot \tfrac{1}{n-1}$. The fourth condition occurs with probability at least $\left(\frac{n-2}{n} \cdot \frac{n-3}{n-1}\right)^n$. Finally, the fifth condition occurs with probability $\tfrac{1}{2}$. For sufficiently large~$n$, this probability can be lower bounded by
    \[
        \frac{p^2}{2^3} \cdot \frac{1}{n} \cdot \left(\frac{n-2}{n} \cdot \frac{n-3}{n-1}\right)^n \geq \frac{p^2}{2^4e^4(n-1)}.
    \]

    The number of ordered pairs of students is $n(n-1)$. List all such pairs and let the $j$th element of the list be $I_j$. Let $\mathbbm{1}_j$ be the indicator random variable which represents whether the students in $I_j$ satisfy all of the above conditions. Then the expected number of pairs of students that satisfy all conditions is 
    \[
        \mathbb{E}\left[\sum_{j=1}^{n(n-1)} \mathbbm{1}_ j\right] \geq n(n-1) \cdot \frac{p^2}{2^4e^4(n-1)} = \frac{p^2}{2^4e^4} \cdot n.
    \]
    We observe that by construction, it is impossible for there to be two different pairs $(i_1,i_2)$ and $(i'_1,i'_2)$ that satisfy the above conditions such that $i'_1=i_1$ or $i'_2=i_2$. Indeed, this would imply that $(i_1,i_2)=(i'_1,i'_2)$ because the identity of each of $i_1$ and $i_2$ uniquely determines that of the other, while the identity of each of $i'_1$ and $i'_2$ uniquely determines that of the other in the same way.

    To conclude, all such pairs of students involve distinct students, so there are at least an expected $\frac{p^2}{2^4e^4} \cdot n$ unconstrained students who each strictly prefer for a distinct constrained student to become unconstrained. We  choose $\tau = \frac{p^2}{2^4e^4}$ to satisfy the theorem statement.
\end{proof}

We now demonstrate the use of the second, slightly more elaborate, construction for proof of a case of \cref{constraint}.

\begin{lemma}
    For every $p \in (0, 1)$, there exists $\tau > 0$ such that for any large enough $n$, in the uniform $(n; 2)$ model with category-probability lower bound $p$ and with the districts using any combination of matching mechanisms, there exists a set of constrained students of expected size at least $\tau n$ where each constrained student strictly prefers for a distinct constrained student \uline{in another district} to become unconstrained.
\end{lemma}
\begin{proof}
    We start by lower bounding the expected number of ordered trios of students $(i_1, i_2, i_3)$ that satisfy the following conditions:
    \begin{enumerate}
        \item $i_1$ is constrained and resides in district $L$, $i_2$ is constrained and resides in district $R$, and $i_3$ is unconstrained.
        \item $i_1$ most prefers a school $\ell_1 \in L$.
        \item $i_2$ most prefers a school $\ell_2 \in L$, and second most prefers a school $r_1 \in R$.
        \item $i_3$ most prefers $r_1$, and second most prefers $\ell_1$.
        \item No student finds a school in $\{\ell_1, \ell_2, r_1\}$ acceptable except as above.
        \item $i_3$ is higher in the priority order than $i_1$ at $\ell_1$, and $i_2$ is higher in the priority order than $i_3$ at $r_1$.
    \end{enumerate}
    Under these conditions, $i_1$ strictly benefits from $i_2$ becoming unconstrained, as $i_2$ then switches from enrolling in district $R$ to enrolling in district $L$. As a result, $i_3$ switches from enrolling in district $L$ to enrolling in district $R$, which opens up school $\ell_1$ for~$i_1$.
        
    We lower bound the probability that this set of conditions occurs for a specific $(i_1, i_2, i_3)$. Conditioned on all previous conditions, the first condition occurs with probability at least~$p^3$, the second condition occurs with probability at least $\tfrac{1}{2}$, the third condition occurs with probability $(\tfrac{n-1}{n} \cdot \tfrac{1}{2} \cdot \tfrac{n-2}{n-1} \cdot \tfrac{1}{2})$, and the fourth condition occurs with probability $(\tfrac{1}{n} \cdot \tfrac{1}{n-1})$. The fifth condition occurs with probability at least $\left(\frac{n-3}{n} \cdot \frac{n-4}{n-1}\right)^n$. Finally, the sixth condition occurs with probability $\tfrac{1}{2^2}$. For sufficiently large $n$, this probability can be lower bounded by
    \[
        \frac{p^3}{2^5} \cdot \left(\frac{n-1}{n} \cdot \frac{n-2}{n-1}\right) \cdot \left(\frac{1}{n} \cdot \frac{1}{n-1}\right) \cdot \left(\frac{n-3}{n} \cdot \frac{n-4}{n-1}\right)^n \geq \frac{p^3}{2^6e^{6}(n-1)(n-2)}.
    \]

    The number of ordered trios of students is $n(n-1)(n-2)$. List all such trios and let the $j$th element of the list be $I_j$. Let $\mathbbm{1}_j$ be the indicator random variable that represents whether the students in $I_j$ satisfy all of the above conditions. Then the expected number of trios of students that satisfy all conditions is 
    \[
        \mathbb{E}\left[\sum_{j=1}^{n(n-1)(n-2)} \mathbbm{1}_{j}\right] \geq n(n-1)(n-2) \cdot \frac{p^3}{2^6e^{6}(n-1)(n-2)} = \frac{p^3}{2^6e^{6}} \cdot n.
    \]
    We observe that by construction, it is impossible for there to be two different trios $(i_1,i_2,i_3)$ and $(i'_1,i'_2,i'_3)$ that satisfy the above conditions such that $i'_1=i_1$ or $i'_2=i_2$. Indeed, this would imply that $(i_1,i_2,i_3)=(i'_1,i'_2,i'_3)$ because the identity of each of $i_1$ and $i_2$ uniquely determines that of $i_3$, which in turn uniquely determines those of all students in $(i_1,i_2,i_3)$, while the identity of each of $i'_1$ and $i'_2$ uniquely determines that of $i'_3$ and hence those of all students in $(i'_1,i'_2,i'_3)$ in the same way.

    To conclude, all such trios of students involve distinct ``$i_1$'' and ``$i_2$'' students, so there are at least an expected $\frac{p^3}{2^6e^{6}} \cdot n$ constrained students who each strictly prefer for a distinct constrained student in another district to become unconstrained. We  choose $\tau = \frac{p^3}{2^6e^{6}}$ to satisfy the theorem statement.
\end{proof}

\section{Proofs Omitted from Section~\ref{sec:sincere-unconstrained}}\label{app:heuristics}

The full details of the example described in \cref{sec:heuristics} are as follows:

\noindent\begin{minipage}[t]{.5\linewidth}
\vspace{0.5pt}
\hfil Student preferences \hfil
\hrule
\begin{align*}
    \text{(sincere) }i_1: & \ell_2 \succ \ell_3 \succ r_2 \succ r_1 \succ \ell_1 \\
    \text{(sophisticated) }i_2: & \ell_1 \succ r_3 \\
    \text{(sincere) }i_3: & r_2 \\
    \text{(sincere) }i_4: & r_1 \succ r_4 \\
    \text{(sincere) }i_5: & \ell_2 \\
    \text{(sincere) }i_6: & \ell_3 \\
\end{align*}
\end{minipage}%
\begin{minipage}[t]{.5\linewidth}
\vspace{0.5pt}
\hfil School priority orderings \hfil
\hrule
\begin{align*}
    \ell_1: & i_1 - i_2 \\
    \ell_2: & i_5 - i_1 \\
    \ell_3: & i_6 - i_1 \\
    r_1: & i_1 - i_4\\
    r_2: & i_3 - i_1 \\
    r_3: & i_2 \\
    r_4: & i_4
\end{align*}
\vspace{0.5pt}
\end{minipage}

\begin{proof}[Proof of \cref{thm:positional_heuristics}]
    As in the proof of Theorem \ref{thm:sophisticated}, we start by lower bounding the expected number of ordered sextets of students that satisfy the following conditions. 
    \begin{enumerate}
        \item \textbf{Conditions on sophistication types, constraint types, and residence:} 
        \begin{enumerate}
            \item $i_1$ is sincere and unconstrained.
            \item $i_2$ is sophisticated and unconstrained.
            \item Both $i_3$ and $i_4$ are sincere and reside in $R$.
            \item Both $i_5$ and $i_6$ are sincere and reside in $L$.
        \end{enumerate}
        \item \textbf{Conditions on student preferences and school locations:}
        \begin{enumerate}
            \item $i_6$ most prefers a school $\ell_3 \in L$.
            \item $i_5$ most prefers a school $\ell_2 \in L$.
            \item $i_4$ most prefers a school $r_1 \in R$ and second most prefers a school $r_4\in R$.
            \item $i_3$ most prefers a school $r_2 \in R$.
            \item $i_2$ most prefers a school $\ell_1 \in L$ and second most prefers a school $r_3 \in R$.
            \item $i_1$ has preference ordering $\ell_2 \succ \ell_3 \succ r_2 \succ r_1 \succ \ell_1$.
            \item No student finds a school in $\{\ell_1,\ell_2,\ell_3,r_1,r_2,r_3,r_4\}$ acceptable except as above. 
        \end{enumerate}
        \item \textbf{Conditions on school priorities:} 
        \begin{enumerate}
            \item $i_1$ has priority over $i_2$ at $\ell_1$.
            \item $i_5$ has priority over $i_1$ at $\ell_2$.
            \item $i_6$ has priority over $i_1$ at $\ell_3$.
            \item $i_1$ has priority over $i_4$ at $r_1$.
            \item $i_3$ has priority over $i_1$ at $r_2$.
        \end{enumerate}
    \end{enumerate}
    Our analysis changes as follows. We added two new students, who are sincere and reside in $L$, and so the conditions on sophistication types, constraint types, and residence are satisfied with probability at least $p^6$. We also added two new schools, which must each prefer a specific student over $i_1$, so the conditions on school priorities are satisfied with probability $\tfrac{1}{2^5}$. The conditions on student preferences and school locations are satisfied with probability at least
    \begin{multline*}
        \left(\frac{1}{2}\right) \cdot \left(\frac{n-1}{n} \cdot \frac{1}{2}\right) \cdot \left(\frac{n-2}{n} \cdot \frac{1}{2} \cdot \frac{n-3}{n-1} \cdot \frac{1}{2}\right) 
        \cdot \left(\frac{n-4}{n}\cdot\frac{1}{2}\right)\cdot \left(\frac{n-5}{n} \cdot \frac{1}{2} \cdot \frac{n-6}{n-1}\cdot \frac{1}{2} \right) \cdot \\* \cdot \left(\frac{1}{n(n-1)(n-2)(n-3)(n-4)}\right) \cdot \left(\frac{n-7}{n} \cdot \frac{n-8}{n-1} \cdot \frac{n-9}{n-2} \cdot \frac{n-10}{n-3} \cdot \frac{n-11}{n-4} \right)^{n}.
    \end{multline*}
    For sufficiently large $n$, this probability can be lower bounded by    
    \[
        \frac{1}{2^7} \cdot \frac{(n - 5)!}{n!} \cdot \frac{1}{e^{35}} \cdot \frac{1}{2}.
    \]
    The number of ordered sextets of distinct students is $\tfrac{n!}{(n - 6)!}$. The expected number of sextets of students that satisfy all conditions is therefore at least
    \[
        \frac{n!}{(n - 6)!} \cdot p^6 \cdot \frac{1}{2^5} \cdot \frac{1}{2^{8} \cdot e^{35}} \cdot \frac{(n - 5)!}{n!} = \frac{p^{6}(n - 5)}{2^{13} \cdot e^{35}}.
    \]
    The rest of the analysis follows as in the proof of Theorem \ref{thm:sophisticated}. We choose $\tau = \frac{p^{6}}{2^{14} \cdot e^{35}}$ to satisfy the theorem statement.
\end{proof}